\documentclass[journal, lettersize]{IEEEtran}
\usepackage{amsmath,amsfonts}
\usepackage{amsthm}
\usepackage{amssymb}
\usepackage{array}
\usepackage{cancel}
\usepackage[caption=false,font=normalsize,labelfont=sf,textfont=sf]{subfig}
\usepackage{textcomp}
\usepackage{stfloats}
\usepackage{url}
\usepackage{verbatim}
\usepackage{graphicx}
\usepackage{cite}
\usepackage{upgreek}
\usepackage{mathtools}
\usepackage{multirow}
\usepackage{multicol}
\usepackage{numprint}
\usepackage{dsfont}
\usepackage{tikz}
\usepackage{xfrac}

\newtheorem{theorem}{Theorem}
\newtheorem{lemma}{Lemma}
\newtheorem{proposition}{Proposition}

\usepackage{listings,xfp}
\usepackage{anyfontsize}

\npdecimalsign{\ensuremath{.}}

\newcommand{\sign}{\ensuremath{\operatorname{sign}}}
\newcommand{\erf}{\ensuremath{\operatorname{erf}}}

\DeclarePairedDelimiter\floor{\lfloor}{\rfloor}

\newcommand{\mybinom}[3][0.8]{\scalebox{#1}{$\dbinom{#2}{#3}$}}

\makeatletter
{\small 
\xdef\f@size@small{\f@size}
\xdef\f@baselineskip@small{\f@baselineskip}
\normalsize 
\xdef\f@size@normalsize{\f@size}
\xdef\f@baselineskip@normalsize{\f@baselineskip}
}
\newcommand{\smalltonormalsize}{%
  \fontsize
    {\fpeval{(\f@size@small+\f@size@normalsize)/2}}
    {\fpeval{(\f@baselineskip@small+\f@baselineskip@normalsize)/2}}%
  \selectfont
}
\makeatother

\makeatletter
{\footnotesize 
\xdef\f@size@footnotesize{\f@size}
\xdef\f@baselineskip@footnotesize{\f@baselineskip}
\small 
\xdef\f@size@small{\f@size}
\xdef\f@baselineskip@small{\f@baselineskip}
}
\newcommand{\footnotesizetosmall}{%
  \fontsize
    {\fpeval{(\f@size@footnotesize+\f@size@small)/2}}
    {\fpeval{(\f@baselineskip@footnotesize+\f@baselineskip@small)/2}}%
  \selectfont
}
\makeatother

\begin{document}

\title{On the Impact of Dynamic Beamforming on EMF Exposure and Network Coverage: A Stochastic Geometry Perspective}

\author{Quentin~Gontier,~\IEEEmembership{Graduate Student Member,~IEEE,} Charles~Wiame,~\IEEEmembership{Member,~IEEE,} 
Joe~Wiart,~\IEEEmembership{Senior Member,~IEEE,} François~Horlin,~\IEEEmembership{Member,~IEEE,} Christo~Tsigros, Claude~Oestges,~\IEEEmembership{Fellow,~IEEE,} and  Philippe~De~Doncker,~\IEEEmembership{Member,~IEEE}
\thanks{This work was supported by Innoviris under the Stochastic Geometry Modeling of Public Exposure to EMF (STOEMP-EMF)~grant.}
\thanks{Q. Gontier, F. Horlin and Ph. De Doncker are with Universit\'e Libre de Bruxelles, OPERA-WCG, Avenue Roosevelt 50 CP 165/81, 1050 Brussels, Belgium (quentin.gontier@ulb.be).} 
\thanks{C. Wiame is with NCRC Group, Massachusetts Institute of Technology, Cambridge, MA 02139 USA.} 
\thanks{J. Wiart is with Chaire C2M, LTCI, Télécom Paris, Institut Polytechnique de Paris, 91120 Palaiseau, France.} 
\thanks{C. Tsigros is with Department Technologies et Rayonnement, Brussels Environment, Belgium.} 
\thanks{C. Oestges is with ICTEAM Institute, Université Catholique de Louvain, Louvain-la-Neuve, Belgium.} 
}

\maketitle

\begin{abstract}
This paper introduces a new mathematical framework for dynamic beamforming-based cellular networks, grounded in stochastic geometry. The framework is used to study the electromagnetic field exposure (EMFE) of active and idle users as a function of the distance between them. A novel multi-cosine antenna pattern is introduced, offering more accurate modeling by incorporating both main and side lobes. Results show that the cumulative distribution functions of EMFE and coverage obtained with the multi-cosine pattern align closely with theoretical models, reducing error to less than 2\%, compared to a minimum of 8\% for other models. The marginal distribution of EMFE for each user type is mathematically derived. A unique contribution is the introduction of the SCAIU (\underline{S}patial \underline{C}DF for \underline{A}ctive and \underline{I}dle \underline{U}sers), a metric that ensures coverage for active users while limiting EMFE for idle users. Network performance is analyzed using these metrics across varying distances and antenna elements. The analysis reveals that, for the chosen network parameters, with 64 antenna elements, the impact on idle user EMFE becomes negligible beyond 60~m. However, to maintain active user SINR above 10 dB and idle user EMFE below -50~dBm at 2~m, more than 256 elements are required.
\end{abstract}


\section{Introduction}

\IEEEPARstart{T}{he} rapid evolution of wireless communication technologies has sparked growing concerns regarding electromagnetic field exposure (EMFE). 
Entities such as the International Commission on Non-Ionizing Radiation Protection (ICNIRP) establish maximal EMFE thresholds based on conservative margins and literature review, specifying basic restrictions in terms of specific absorption rate (SAR) or incident power density (IPD) \cite{icnirp2020}. 

An important innovation introduced by the 5th generation of cellular networks is dynamic beamforming (DBF). Using multiple antennas at the base station (BS), DBF enables the formation of narrow dynamic beams. While this significantly enhances the signal-to-interference-plus-noise ratio (SINR), it results in higher EMFE for \textit{active users} (AUs) calling for the beam, compared to \textit{idle users} (IUs) who are not active in the network \cite{ANFR2019a}. Due to these differences, several authors and regulatory bodies advocate for the establishment of distinct EMFE limits for AUs and IUs, with the aim of creating "reduced EMFE areas" in sensitive locations such as hospitals, schools, or public buildings like train stations \cite{strinati2021wireless}. Investigating the EMFE experienced by IUs, especially its variation with distance from an AU \cite{IMEC_nl}, is crucial in this context. 

The EMFE experienced by IUs is influenced by contributions from close BSs serving adjacent AUs, but also from BSs located further away, with the dynamic beam's shape playing a pivotal role. This complexity underscores the interconnected nature of IU's EMFE, minimal SINR requirements, and maximum EMFE allowances for AUs, influencing network design strategies. To comprehensively analyze these factors, this paper explores global network performance by simultaneously examining EMFE and coverage for AUs, alongside EMFE for IUs through the SCAIU (\underline{S}patial \underline{C}DF for \underline{A}ctive and \underline{I}dle \underline{U}sers) metric.

The SCAIU metric involves random parameters, including the network topology, BS beam directions and propagation channel. Stochastic geometry (SG) emerges as a powerful tool for capturing this inherent randomness, modeling BSs as point processes (PPs) to formulate network performance in mathematically and computationally tractable integrals \cite{Lee2013}.

Motivated by these considerations, this paper aims to introduce a comprehensive mathematical framework for marginally and jointly evaluating EMFE for both AUs and IUs and the coverage of AUs, utilizing an antenna gain model that closely approximates real-world conditions.

\subsection{Related Works}
\label{ssec:related_works}

\paragraph{Exploring EMFE}
Numerous studies have evaluated EMFE in 5G and beyond networks, where the complexity of heterogeneous networks includes macro, small, and femto cells, higher frequencies than for previous generations of cellular networks, and active antenna arrays. Studies on 5G EMFE range from field measurements to simulation and mathematical models. Urban measurements under user calls and idle states using a specific protocol have been validated \cite{IMEC_protocol_paper}. Simulation studies by the French spectrum regulator ANFR have revealed increased EMFE for AUs and reduced EMFE for IUs following the deployment of 5G \cite{ANFR2021}. Comprehensive reviews on the current state of EMFE evaluation for 5G BSs are available in \cite{Elbasheir21} and \cite{patsouras_2023_8099834}, highlighting the reduction in EMFE levels through DBF.

Strategies to mitigate EMFE include optimization algorithms, smart power control schemes \cite{Ajibare21}, and reconfigurable intelligent surfaces (RIS) for creating low-EMFE zones while maintaining data rates \cite{Ibraiwish22}.

But challenges remain; in-situ measurements only provide insights at specific locations under particular conditions, while deterministic numerical evaluations struggle to efficiently capture all sources of randomness within the network.

\paragraph{DBF models in SG}
SG is pivotal for computing performance metrics in large networks with random parameters. BSs are often modeled as homogeneous Poisson point processes (PPPs) to balance accuracy and computational tractability. SG has extensively studied aspects like SINR and ergodic data rate \cite{Lee2013}, including DBF.

Antenna pattern models approximate the array factor of a Uniform Linear Array (ULA), which is mathematically intractable. The sectored or flat-top pattern, widely used despite its simplifications, aids in modeling SINR in multi-tier millimeter-wave (mmWave) networks and beam misalignment scenarios \cite{Kalamkar22}. The truncated cosine approximation of the main lobe offers closer simulation results to theoretical patterns, benefiting various applications \cite{yu2017, Nabil23}. The Gaussian approximation models the side lobes more realistically, which is crucial for studies on beam misalignment in mmWave networks \cite{Rebato19}. Novel approaches like cylindrical arrays explore interference in non-traditional configurations \cite{math10071156}.

\paragraph{EMFE assessment using SG}

SG has been instrumental in assessing and optimizing IPD for Wireless Power Transfer (WPT) systems and EMF-aware applications. In SG studies of DBF, the flat-top model supports energy harvesting and correlation analyses \cite{ECC} and is used for simultaneous wireless information and power transfer applications \cite{SWIPT_MIMO}. Additionally, the Gaussian model is used to analyze imperfect beam alignment in mmWave WPT systems \cite{Wang21}.

Early works on EMFE assessment within the SG framework include a comparison between theoretical and experimental EMFE distributions \cite{GontierAccess} and an analysis of max-min fairness power control. The flat-top pattern is used for EMFE analysis in coexisting sub-6~GHz and mmWave networks \cite{9511258}, and for a joint analysis of EMFE and SINR in $\beta$-Ginibre PP and inhomogeneous Poisson PPs in \cite{GontierTWC}. Notably, these two works are the soles that incorporate DBF into the SG model for EMFE. Joint analyses of EMFE and SINR are explored in specific scenarios in downlink (DL) alone \cite{manhattan, cell-free} or in DL and uplink \cite{gontier2023uplink, chen2023joint}. 

However, literature gaps persist in studying EMFE for idle users in DBF networks and accurately modeling antenna main and side lobes. 

\subsection{Contributions}
\label{ssec:contributions}
The key contributions of this paper are outlined below:
\begin{enumerate}
    \item \textit{Enhanced Antenna Modeling for EMFE Analysis}: This paper introduces the multi-cosine antenna gain model, which provides a more accurate representation of EMFE by accounting for side lobe effects. Mathematical formulations are presented for the marginal cumulative distribution function (CDF) characterizing the EMFE experienced by IUs. The results show that the multi-cosine model produces CDFs that align more closely with the theoretical antenna pattern than those obtained using existing models, thereby offering deeper insights into the differential EMFE experienced by AUs and IUs across the network.
    \item \textit{Novel spatial performance metric}: Beyond EMFE, the study expands to analyze broader network performance. This includes computing the SCAIU metric which assesses the probability of meeting SINR requirements for AUs while constraining EMFE for IUs below specified thresholds.
    \item \textit{Network management}: The study explores network performance across varying numbers of antenna elements and distances between IUs and AUs. This analysis identifies optimal configurations for achieving specified network performance levels, providing comprehensive insights into the dynamic aspects of network management.
\end{enumerate}

\section{System Model}
\label{sec:system_model}

\subsection{Topology}
\label{ssec:topology}
Consider a two-dimensional spatial domain $\mathcal{B} \in \mathbb{R}^2$ representing the network area, defined as a disk with radius $\tau$ centered at the origin. Within $\mathcal{B}$, let $\Psi = {X_i}$ denote the PPP representing the locations of BS $X_i$, all sharing the same technology, belonging to the same network provider, operating at a carrier frequency $f$, and being able to transmit at a effective isotropic radiated power (EIRP) $P_t$. The density of BSs in $\Psi$ is denoted by $\lambda$. Each BS is located at a height $z > 0$ relative to the users. The AU is positioned at the origin and is served by the nearest BS, while all other BSs act as potential interferers. An IU is located at a distance $d << \tau$ from the AU. The angle formed by the IU and $X_i$ at the AU's location is denoted as $\Theta_i$, as depicted in Fig.~\ref{fig:Active and idle user}. It is positive if measured in the trigonometric direction and negative if measured clockwise. The distance between the AU and $X_i$ is denoted as $R_i$ and the distance between the IU and $X_i$ as $W_i \! = \!\sqrt{R_i^2+d^2-2R_i d \cos(\Theta_i)}$. Additionally, $\delta_i \! = \! \sign(\Theta_i) \arccos\left((R_i - d \cos \Theta_i)/W_i\right)$ describes the angle between the AU and the IU from the perspective of $X_i$. The boresight direction of $X_i$ forms an angle $\xi_i$ with the AU, with $X_0$ being the serving BS and $\xi_0 = 0$.
\begin{figure}
    \centering
    \includegraphics[scale = 0.5, trim={2cm, 0cm, 0cm, 1cm}, clip]{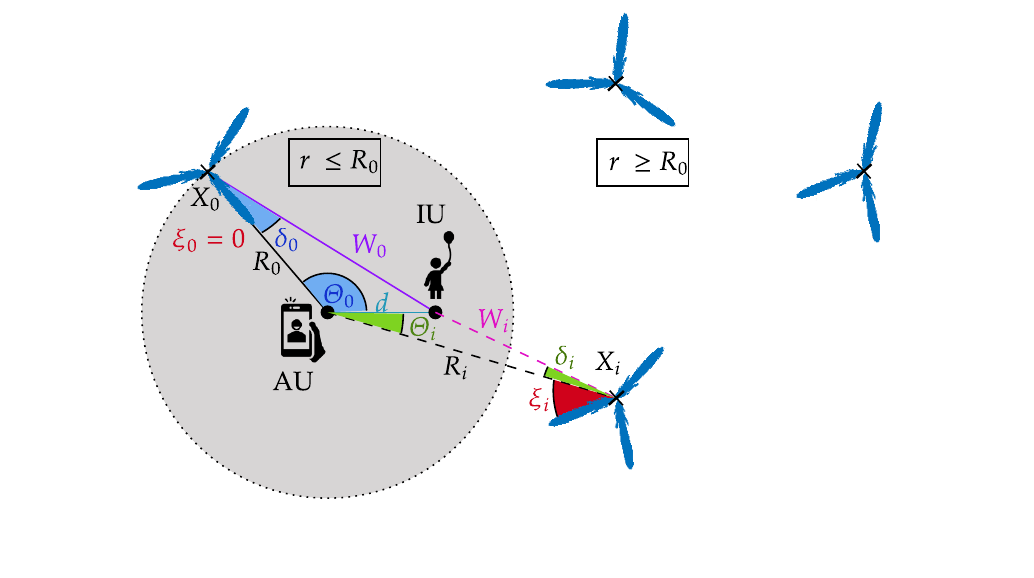}
    \caption{Scheme of the network with an AU at the origin and an IU at a distance $d$. The antenna pattern of a typical BS is shown in the lower right corner ($N \! = \! 16$).}
    \label{fig:Active and idle user}
\end{figure}  

In this configuration, the equipment of the AU is considered with a unit gain. Each BS is equipped with three identical ULAs oriented at 120$^\circ$ intervals, following 3GPP specifications, each array containing $N$ antenna elements with half-wavelength spacing. There is no intracell interference, which means that ULAs of the same BS do not interfere with each other. Consequently, the angle $\xi_i$ is a random variable uniformly distributed over $[-\pi/3,\pi/3[$. An exclusion radius $r_e$ around any user ensures no BS is located within this region. The normalized BS gain $G(\xi)$ is uniformly normalized by the maximum gain $G_{\textrm{\normalfont max}} \! = \! N$, and defined in Subsection~\ref{ssec:gain}. The network is operating at full capacity, with each ULA communicating with one AU.

For a PPP, the probability density function (PDF) of $R_0$ is given by 
\begin{equation}\label{eq:pdf_R0}
    f_{R_0}(r) = \frac{2 \pi \lambda r \exp(-\lambda \pi r^2)}{\exp(-\lambda \pi r_e^2)-\exp(-\lambda \pi \tau^2)} \quad \text{ for } r > r_e.
\end{equation}

To ensure the significance of $d$, the IU must be within the same cell as the AU. Hence, subsequent analyses assume $d$ is smaller than the mean cell radius, defined as $(2 \sqrt{\lambda})^{-1}$. If $d$ exceeds this threshold, the user is categorized as a \textit{random users} (RUs), implying no correlation with the AU's location.

\subsection{Propagation Model}
\label{ssec:propagation}

The propagation model is defined as
{\smalltonormalsize
\begin{equation}\label{eq:model}
    P_{r, i} = P_t G_i |h_i|^2 l_i 
\end{equation}}
where $P_{r, i}$ is the received power from BS $X_i$, $G_i$ is the BS gain towards the user (AU or IU), $|h_i|^2$ accounts for the fading and $l_i$ is the channel power gain due to path loss. Specifically, $l_i \! = \! l(X_i) \! = \! \kappa^{-1}\left(R_i^2+z^2\right)^{-\alpha/2}$ for a distance $R_i$ between $X_i$ and the AU, with $\alpha > 2$ the path loss exponent and $\kappa \! = \! (4\pi f/c)^2$ where $c$ is the speed of light. For the path between $X_i$ and the IU, $\Tilde{l}(X_i) \! = \! \kappa^{-1} \left(W_i^2+z^2\right)^{-\alpha/2}$. The fading $h_i$ follows a Nakagami-$m$ model, making $|h_i|^2$ gamma-distributed with shape parameter $m$ and scale parameter $1/m$. Consequently, the CDF of $|h_i|^2$ is expressed as $F_{|h_i|^2}(x) \! = \! \gamma(m,m x)/\Gamma(m)$.

Define $\Bar{P}_{r,i}\! =\! \Bar{P}_{r}(r_i) \!= \!P_t l_i(r_i)$. Since $G_0(0) = 1$, let $S_0(0)\! =\! \Bar{P}_{r}(r_0)|h_0|^2$ be the useful power received by the AU from $X_0$ and let $I_0(0) \! = \! \sum_{i \in \Psi \setminus\left\{X_0\right\}} \Bar{P}_{r,i}(r_i) G_{i}(\xi) |h_i|^2$ be the aggregate interference at the AU's location. Similarly, the signal coming from $X_0$ and reaching the IU is $S_0(d)\! =\! \Bar{P}_{r}(W_0)G(\delta_0)|h_0|^2$ and the aggregate interference at the IU's location is $I_0(d) \! = \! \sum_{i \in \Psi \setminus\left\{X_0\right\}} \Bar{P}_{r,i}(W_i) G_{i}(\xi_i+\delta_i) |h_i|^2$. Based on these definitions, the SINR experienced by the AU and conditioned on the distance to the serving BS is given by
{\smalltonormalsize
\begin{align}\label{eq:SINR}
    \text{\normalfont{SINR}}_0 = \frac{S_0}{I_0 + \sigma^2}
\end{align}}where $\sigma^2 \!= \! k\,B^{w}\,T\,\mathcal{F}^w$ is the thermal noise power. In the following, the performance metrics will be derived for the DL power EMFE defined as
{\smalltonormalsize
\begin{align}\label{eq:expP}
    \mathcal{P} = \sum_{i \in \Psi} \Bar{P}_{r,i}\,G_{i}\, |h_i|^2 = S_0+I_0,
\end{align}}which can be converted into a total IPD as, by definition,
{\smalltonormalsize
\begin{equation}\label{eq:ExpWM2}
    \mathcal{S} = \sum_{i \in \Psi} \frac{P_{t}\,G_{i}\,|h_i|^2}{4\pi \left(r_i^2+z^2\right)^{\alpha/2}} = \frac{\kappa}{4\pi}\mathcal{P}.
\end{equation}}

\subsection{Antenna Pattern Models}
\label{ssec:gain}
The various antenna patterns considered here can be observed in Fig.~\ref{fig:Beamforming_gain_theta_zoom3}, with a zoom on the first side lobes. The normalized gain of one vertical ULA with $N$ omnidirectional antenna elements and half-wavelength spacing is given~by
\begin{equation}\label{eq:Gact}
    G_{\textrm{ULA}}(\varphi) = \frac{\sin^2\left(\frac{\pi\,N}{2}\,\sin(\varphi)\right)}{N^2\,\sin^2\left(\frac{\pi}{2}\,\sin(\varphi)\right)}
\end{equation} where $\varphi \in [-\pi/3, \pi/3[$ and $\varphi = 0$ corresponds to the maximal normalized gain of the main lobe $G_{\textrm{ULA}}(0) = 1$. While the gain function is closely approximated by a squared cardinal sinus function, this approximation also yields intractable mathematical expressions for calculating performance metrics. Instead, a flat-top antenna pattern is widely used in the literature \cite{Kalamkar22}, and given by
\begin{equation}\label{eq:Gft}G_{ft}(\varphi) = \left\{
    \begin{array}{ll}
    1   \quad  & \text{if  } |\varphi|\leq \varphi_{3\text{dB}} \\
    g    \quad &  \text{otherwise}
    \end{array}\right.
\end{equation} where $\varphi_{3\text{dB}}$ is half of the half-power beamwidth (HPBW) of the theoretical ULA pattern and $g$ is the chosen side lobe gain.

The truncated cosine antenna pattern \cite{yu2017} approximates the main lobe of the theoretical pattern while assuming null gain for the side lobes and is expressed as
\begin{equation}\label{eq:Gcos}G_{cos}(\varphi) = \left\{
    \begin{array}{ll}
    \cos^2\left(N \pi \varphi/4\right)   \quad  & \text{if  } |\varphi|\leq 2/N\\
    0    \quad &  \text{otherwise.}
    \end{array}\right.
\end{equation}

The Gaussian approximation is given by
\begin{equation}\label{eq:GG}
\end{equation}where $\varphi \in [-\pi/3, \pi/3[$ and $\eta = \ln\left(\frac{N-g}{N/2-g}\right)/\varphi_{3\text{dB}}^2$.

Lastly, we introduce the new multi-cosine antenna pattern model defined as
{\smalltonormalsize
\begin{equation}\label{eq:Gmc}G_{mc}(\varphi) = \left\{
    \begin{array}{ll}
    \!\cos^2\!\left(\frac{N \pi \varphi}{4}\right)  \!   & \text{if  } |\varphi|\leq 2/N\\
    \!\chi_k \sin^2\!\left(\frac{N \pi \varphi}{2}\right)  \!   & \text{if  } \frac{2 k}{N}\!\leq\!|\varphi|\!\leq \!\frac{2 k+2}{N}
    \end{array}\right.
\end{equation}} where $\chi_k = \frac{\sin^2(N x_k)}{N^2 \sin^2(x_k)}$ is the extrema of the $k$th side lobe of the theoretical gain function, with $0 \leq k \leq k_{\textrm{\normalfont max}}$ and $\chi_0\! = \!1$. The choice of $k_{\textrm{\normalfont max}}$ is flexible but should remain below $\floor*{N \sqrt{3}/4-1}$ to prevent side lobes from extending beyond each ULA's sector. The values of $x_k$ are well approached by the ordered positive solutions of $N \tan(x) = \tan(N x)$.

\begin{figure}
\centering
\begin{minipage}{.5\linewidth}
  \centering
\includegraphics[width=0.95\linewidth, trim={6.5cm, 9cm, 7cm, 10cm}, clip]{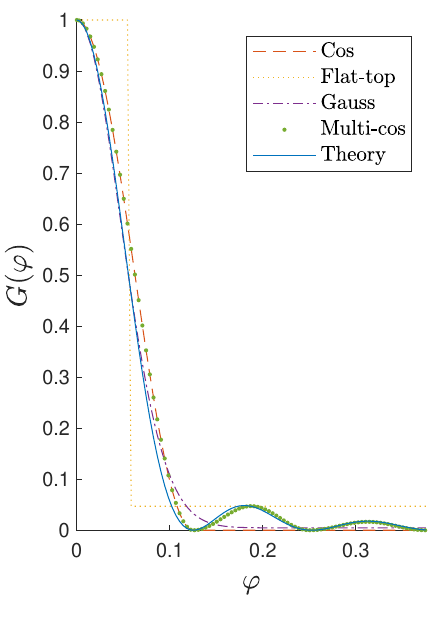}
\end{minipage}%
\begin{minipage}{.5\linewidth}
  \centering
\includegraphics[width=0.95\linewidth, trim={6.5cm, 9cm, 7cm, 10cm}, clip]{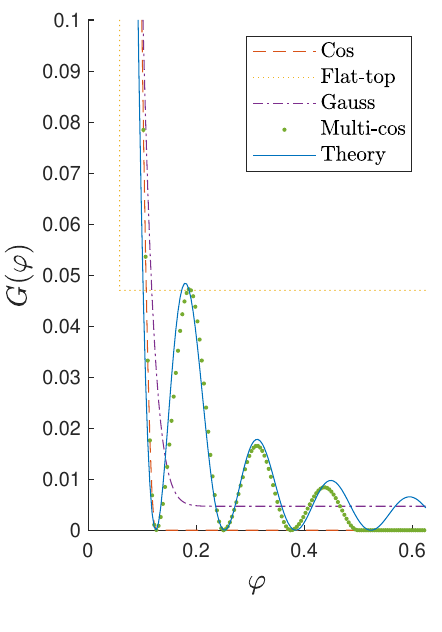}  
\end{minipage}
\caption{Antenna patterns for positive angles and zoom ($N \! = \! 16$)}
\label{fig:Beamforming_gain_theta_zoom3}
\end{figure}

\section{Mathematical Framework}
\label{sec:analytical_results}


\subsection{Preliminaries}
\label{ssec:prelim}
Our analysis begins with the computation of the central moments of the gain models for each BS. 

\begin{proposition}\label{prop:gain}
    The $k$th moments ($k>0$) of the gain models are given by
    \begin{equation}\label{eq:mom_Gft}
        \mathbb E\left[G_{ft}^k(\varphi)\right] = \frac{3}{\pi} \varphi_{3\textrm{\normalfont dB}} (1-g^k) + g^k;
    \end{equation}
    \begin{equation}\label{eq:mom_Gcos}
        \mathbb E\left[G_{cos}^k(\varphi)\right] = \dfrac{6 \,\Gamma(k+1/2)}{N {\pi}^{3/2} \Gamma(k+1)};
    \end{equation}
    \begin{equation}\label{eq:mom_GG}
        \mathbb E\left[G_{G}^k(\varphi)\right] = g^k + \frac{3}{2} \sum\limits_{p = 1}^k  \binom{k}{p} (N-g)^p g^{k-p} \frac{\erf\left(\tfrac{\pi \sqrt{p \eta}}{3}\right)}{\sqrt{\pi p \eta}};
    \end{equation}
    \begin{equation}\label{eq:mom_Gmc}
        \mathbb E\left[G_{mc}^k(\varphi)\right] = \dfrac{6 \,\Gamma(k+1/2)}{N {\pi}^{3/2} \Gamma(k+1)} (1+\chi_{i_{\textrm{\normalfont max}}}^\dagger)
    \end{equation}
where $\erf(\cdot)$ is the error function and $\chi_{i_{\textrm{\normalfont max}}}^\dagger = \sum\limits_{i = 1}^{i_{\textrm{\normalfont max}}} \chi_i$.
\end{proposition}
\begin{proof}
    The proof is obtained by integrating $\frac{3}{2\pi}\int_{-\pi/3}^{\pi/3} \!G^2(\varphi) d\varphi$. The factor of 3 arises because there are three identical ULAs, each with an equal probability of falling into one of their sectors.
\end{proof}

To derive network performance metrics, the characteristic function (CF) of the useful signal and interference must be calculated at the AU's and the IU's location. The CF of the useful signal is given in Proposition~\ref{prop:cf_S}.
\begin{proposition}\label{prop:cf_S}
    The CF of the useful signal for the propagation model \eqref{eq:model}, from the AU's point of view, conditioned on the distance to the nearest BS $R_0$, is
    \begin{equation}\label{eq:cfS0}
        \phi_S(q;0|R_0) = \mathbb E_{S_0}\left[e^{jq {S_0(R_0)}}\right] = \left(1-jq \Bar{P}_r(R_0)/m\right)^{-m}.
    \end{equation}
    From the IU's point of view, it is given by
    \begin{equation}\label{eq:cfSidle}
    \begin{split}
        \phi_S(q;d|X_0) = \left(1-jq \Bar{P}_r(W_0) G(\delta_0)/m\right)^{-m} \mathds 1\left[|\delta_0| \leq \pi/3\right] \\
        + \eta_S(q|X_0) (1-\mathds 1\left[|\delta_0| \leq \pi/3\right]).
    \end{split}
    \end{equation}
    where $\eta_S(q|X_0)$ depends on the considered gain pattern:
    \begin{itemize}
        \item {Flat-top pattern}
    \end{itemize} 
{\small
\begin{equation}\label{eq:eta_S_ft}
    \begin{split}
        \eta_S(q|X_0) =  \frac{3 \varphi_{3\textrm{\normalfont dB}}}{\pi}\left(1-jq \Bar{P}_r(W_0)/m\right)^{-m}\\
        + \left(1-\frac{3\varphi_{3\textrm{\normalfont dB}}}{\pi}\right) \left(1-jq \Bar{P}_r(W_0) g/m\right)^{-m}
    \end{split}
\end{equation}}

    \begin{itemize}
        \item {Gaussian pattern}
    \end{itemize} 
{\small
\begin{equation}\label{eq:eta_S_G}
\begin{split}
    &\eta_S(q|X_0) =  \left(1-jq \Bar{P}_r(W_0) g/m\right)^{-m}\\
    &+\sum\limits_{p=1}^\infty  \frac{3\erf\left(\tfrac{\pi \sqrt{p \eta}}{3}\right)}{2 \sqrt{\pi p \eta}} (N-g)^p \frac{\left(jq \Bar{P}_r(W_0)\right)^p}{\left(1-\frac{jq \Bar{P}_r(W_0)}{m}\right)^{(m+p)}}\frac{ \Gamma(m+p) }{\Gamma(m) m^p p!} 
\end{split}
\end{equation}}

    \begin{itemize}
        \item {Multi-cos pattern}
    \end{itemize} 
{\small
\begin{equation}\label{eq:eta_S_mc}
        \eta_S(q|X_0) =  1+\frac{6(1+\chi_{i_{\textrm{\normalfont max}}}^\dagger)}{N\pi}  \left(\!{}_2F_1\left(\frac{1}{2}, m;1;\frac{j q \Bar{P}_r(W_0)}{m}\right) -1\right)
\end{equation}}
\end{proposition}
\begin{proof}
    The proof in \eqref{eq:cfS0} is straightforward after applying the expectation operator on $|h|^2$. The proof in \eqref{eq:cfSidle} is given in Appendix~\ref{sec:proofcfS0}.
\end{proof}

The expression of the CF of the interference from the point of view of the IU is assumed to be the same as the one from the point of view of the AU. This is justified by the following observations.
\begin{itemize}
    \item For the AU as well as for the IU, the orientation of the beam of any interfering BS is random.
    \item Since $d << \tau$, it can be assumed that the disk $\mathcal{B}$ centered on the AU coincides with a disk $\mathcal{B}'$ centered on the IU.
    \item Since $d << (2\sqrt{\lambda})^{-1}$, it is assumed that the AU's closest BS is also the IU's closest BS. 
\end{itemize}

\begin{proposition}\label{prop:cf}
    The CF of the interference for the propagation model \eqref{eq:model}, conditioned on the distance to the nearest BS $R_0$, is
    \begin{equation}
        \phi_I(q|R_0) = \mathbb E_{I_0}\left[e^{jq {I_0}}\right] = \exp(-\pi \lambda \eta_I(q|R_0))
    \end{equation}
    where $\eta_I(q|R_0)$ depends on the considered gain pattern:

    \begin{itemize}
        \item {Flat-top pattern}
    \end{itemize} 
{\footnotesize
\begin{equation}
    \begin{split}
        \eta_I(q|R_0) =  \left[r^2-\left(r^2+z^2\right)\frac{3\varphi_{3\textrm{\normalfont dB}}}{\pi} \,_2\!F_1\left(m, -\delta;1-\delta;\frac{j q \Bar{P}_r(r)}{m}\right) \right]_{r = \tau}^{r=r_0}\\
        -\left[\left(r^2+z^2\right)\left(1-\frac{3}{\pi}\varphi_{3\textrm{\normalfont dB}}\right) \,_2\!F_1\left(m, -\delta;1-\delta;\frac{j q \Bar{P}_r(r)g}{m}\right) \right]_{r = \tau}^{r=r_0}\!.
    \end{split}
\end{equation}}

    \begin{itemize}
        \item {Gaussian pattern}
    \end{itemize} 
{\small
\begin{equation}
\begin{split}
    &\eta_I(q|R_0) = \left[r^2-\left(r^2+z^2\right) {}_2F_1(-d,m;1-d;j q \Bar{P}_r(r)g/m)\right.\\
    &\left.+\delta \left(r^2+z^2\right)\sum\limits_{p=1}^\infty  \frac{3\erf\left(\tfrac{\pi \sqrt{p \eta}}{3}\right)}{2 \sqrt{\pi p \eta}} \left(\frac{N}{g}-1\right)^p \frac{\Gamma(p+m)}{\Gamma(m)p!}\right.\\
    &\left.\quad \times B\left(\frac{j q \Bar{P}_r(r)g}{m};p-\delta,1-p-m\right)\left(\frac{j q \Bar{P}_r(r)g}{m}\right)^\delta\right]_{r = r_0}^{r=\tau}.
\end{split}
\end{equation}}
    \begin{itemize}
        \item {Multi-cos pattern}
    \end{itemize} 
{\footnotesize
\begin{equation}
\begin{split}
        &\eta_I(q|R_0) = \frac{3(1+\chi_{i_{\textrm{\normalfont max}}}^\dagger)}{N \pi} \left[\left(r^2+z^2\right)\left(2-2 \,{}_2F_1\left(\frac{1}{2}, m; 1, \frac{j q \Bar{P}_r(r)}{m}\right)\right.\right.\\
        &\quad\left.\left.+\frac{j q \Bar{P}_r(r)}{1-\delta} {}_3\!F_2\left(\frac{3}{2}, 1+m, 1-\delta;2, 2-\delta;\frac{j q \Bar{P}_r(r)}{m}\right)\right) \right]_{r = r_0}^{r=\tau}
\end{split}
\end{equation}}
${}_{p}F_{q}(a_{1},\ldots ,a_{p};b_{1},\ldots ,b_{q};z)=\sum _{n=0}^{\infty }{\frac {(a_{1})_{n}\cdots (a_{p})_{n}}{(b_{1})_{n}\cdots (b_{q})_{n}}}\,{\frac {z^{n}}{n!}}$ is the generalized hypergeometric function with $(a)_k = \Gamma(a+k)/\Gamma(a)$ is the Pochhammer symbol. We use the notation $[f(x)]_{x = a}^{x=b} = f(b)-f(a)$.
\end{proposition}
\begin{proof}
    The proof is similar to the proof of the Laplace transform of the interference in Lemma~2 and Appendix~D in~\cite{yu2017}. We recall the main steps with a few changes in Appendix~\ref{sec:proofcf}.
\end{proof}

\subsection{EMFE}
\label{ssec:efme_cdf}
Using the calculations from the previous section, the CDF of the EMFE for the IU or the AU is provided in Theorem~\ref{th:exp_cdf}.
\begin{theorem}\label{th:exp_cdf}
    The CDF of the EMFE of a user, for the propagation model \eqref{eq:model} in a H-PPP, is given by
\begin{align*}
    &F_{\textrm{EMFE}}^{X}(T_{e}) = \mathbb{P}\left[\mathcal{P}^{X} < T_{e}\right]\\
    &\quad= \frac{1}{2}-\int_{r_e}^{\tau}\int_{0}^{\infty} \frac{1}{\pi q}\textrm{\normalfont Im}\!\left[\phi_E^{X}(q;d|r_0)\,e^{-jqT_{e}}\right]\,dq \, f_R(r_0)\,dr_0.
\end{align*}where $X$ is either the AU with $\phi_E^{AU}(q|r_0)\! = \!\phi_{S}(q;0|r_0)\,\phi_{I}(q|r_0)$ or the IU with $\phi_E^{IU}(q|r_0)\! =\! \frac{1}{2\pi} \int_{0}^{2\pi}\!\phi_{S}(q;d|r_0,\theta)d\theta \,\phi_{I}(q|r_0)$. The EMFE limit is denoted $T_e$.
\end{theorem}
\begin{proof}
    The result follows from the Gil-Pelaez theorem \cite{gil-pelaez}.
\end{proof}
It is worth noting that the CDF of EMFE for a random user is given by setting $\phi_E^{RU}(q;d|r_0)\! = \!\phi_{I}(q|r_e)$ in Theorem~\ref{th:exp_cdf}. The calculation of the expression of the CF of interference differs from the classical method. Instead of applying the expectation operators over $\xi_i$ and $h_i$ first, and then over $r_i$, the method used in this paper, inspired by \cite{yu2017}, first applies the operator over $r_i$ and then over the other variables. This latter approach offers the advantage of easier generalization to various gain and fading models, leveraging knowledge of the moments of $G(\xi)$ and $|h|^2$.

\subsection{Coverage}
\label{ssec:cov_cdf}
The CCDF of the SINR has been derived for models with various features, including DBF with the various described antenna pattern. The difference in this paper lies in the CFs of the signal and interference, whose expression was given in Subsection~\ref{ssec:prelim}. It is therefore given as such in Lemma~\ref{lem:cov}.
\begin{lemma}\label{lem:cov}
    The CCDF of the SINR, for the propagation model \eqref{eq:model} in a H-PPP, is given by
{\smalltonormalsize
\begin{multline*}\label{eq:coveq}
        F_{\text{cov}}(T_c) \triangleq \mathbb E_0\left[\mathbb{P}\left[\text{\normalfont SINR}_0 > T_c\right]\right] \\
        =\int_{r_0}\left(\frac{1}{2} + \int_{0}^{\infty} \text{\normalfont Im}\left[\phi_{\text{\normalfont{SINR}}}(q, T_c|r_0)\right]\,\frac{1}{\pi q} dq \right)f_{R_0}(r_0) dr_0
\end{multline*}}where $\phi_{\text{\normalfont{SINR}}}(q, T_c|r_0) = \phi_S(q;0|r_0)  \phi_I(-T_c q|r_0) e^{-jT_cq\sigma^2}$ and $T_c$ is the SINR threshold.
\end{lemma}

\subsection{SCAIU}
\label{ssec:joint_cdf}
In this subsection, the SCAIU, spatial CDF jointly analyzing the SINR experienced by the AU and the EMFE experienced by the IU, is calculated and given in Theorem~\ref{th:joint1}. The conditional CDF and bounds are then provided.
\begin{theorem}\label{th:joint1}
The joint CDF of the SINR of the AU and the EMFE of the IU, the two users being separated by a distance $d$, for the propagation model \eqref{eq:model} in a H-PPP, is
\begin{align}\begin{split}
    &\mathcal{J}(T_{c}, T_{e};d) = \mathbb E_0\left[\mathbb{P}\left[\textrm{\normalfont SINR}_0 > T_{c}, \mathcal{P}(d) < T_{e}\right]\right]\\
    &\quad= \frac{-1}{4} + \frac{1}{2}\, F_{\textrm{cov}}(T_{c}) + \frac{1}{2}\, F^{IU}_{\textrm{EMFE}}(T_{e}) - \frac{1}{\pi^2}\, \Upsilon(T_{c}, T_{e}; d)
\end{split}\end{align}
where
{\smalltonormalsize
\begin{align}\begin{split}
    &\Upsilon(T_{c}, T_{e}; d) = \frac{1}{2\pi}\int_0^{\tau}\int_0^{2\pi}\Upsilon(T_{c}, T_{e}; d| r_0, \theta_0)\,d\theta \,f_{R_0}(r_0)\,dr_0,\\
    &\Upsilon(T_{c}, T_{e}; d| r_0, \theta_0) = \int_0^{\infty} \!\int_0^{\infty} \!\frac{\epsilon(q, q';T_{c},T_{e},d|r_0, \theta_0)}{q\, q'}\,dq\, dq',\\
    &\epsilon(q, q';T_{c},T_{e},d|r_0, \theta_0) \\
    &\quad \! = \! \frac{1}{2}\,\textrm{\normalfont Re}\left[\epsilon_+(q, q';T_{c},T_{e},d|r_0, \theta_0)-\epsilon_-(q, q';T_{c},T_{e},d|r_0, \theta_0)\right],\\
    &\epsilon_\pm(q, q';T_{c},T_{e},d|r_0, \theta_0) \\
    &\quad= \gamma_\pm(-q T_{c}, q'|r_0)\phi_{S}(q; 0 |\Psi)\phi_{S}(\pm q'; d|\Psi)\, e^{-j (q T_{c}\sigma^2 \pm q' T_{e})}
\end{split}\end{align}}
and $\gamma_\pm(q, q'|r_0)$ is given in~\eqref{eq:gamma} at the top of the next page.
    \begin{figure*}[!h]
{\footnotesize
\begin{multline}\label{eq:gamma}
    \gamma_\pm(q, q'|r_0) = \exp\left(\frac{12 \lambda}{N}\Bigg[(1+k_{\textrm{\normalfont max}})\left(\frac{3(1+k_{\textrm{\normalfont max}})}{N\pi}-1\right)r^2\right.\\    \left.-\sum_{k = 0}^{k_{\textrm{\normalfont max}}}\!\sum_{l = 0}^{m-1}\!\mybinom[0.7]{-1/2}{l}\mybinom[0.7]{m-1}{l}\left(\frac{6}{N\pi}\!\sum_{p = 0}^{k_{\textrm{\normalfont max}}} \!\sum_{l' = 0}^{m-1} \!\mybinom[0.7]{-1/2}{l'}\mybinom[0.7]{m-1}{l'} \frac{r^{2+\alpha(2m-1)}  F_1\!\left(2m\!-\!1\!-\!l\!-\!l'\!+\!\frac{2}{\alpha}, \frac{2m-1}{2}, \frac{2m-1}{2}, 2m\!-\!l-l'\!+\!\frac{2}{\alpha}, \frac{-j m}{ q \chi_k \Bar{P}_r(r)}, \frac{\mp j m}{ q' \chi_k \Bar{P}_r(r)}\!\right)}{\left(2+\alpha(2m-1-l-l')\right)( j q \chi_k \Bar{P}_r(r)/m)^{m-l-\frac{1}{2}}(\pm  j q' \chi_k \Bar{P}_r(r)/m)^{m-l'-\frac{1}{2}}}  \right. \right.\\
    \left.\left.\!+ \!\frac{2{(r^2+z^2)\!\left(1\!-\!\frac{6(1+k_{\textrm{\normalfont max}})}{N\pi}\right)}}{\left(4+\alpha(2m-2l-1)\right)}
    \!\left(\!
    \frac{{}_2F_1\!\left(1, \scalebox{0.9}{1-\textit{l}+}\frac{2}{\alpha}; \frac{1}{2}\scalebox{0.9}{-\textit{l+m}+}\frac{2}{\alpha}; \frac{-j m}{ q \chi_k \Bar{P}_r(r)}\right)}{\left( j q \chi_k \Bar{P}_r(r)/m\right)^{1-l} \left(1- j q \chi_k \Bar{P}_r(r)/m\right)^{m-\frac{3}{2}}}\!+\!\frac{{}_2F_1\!\left(1, \scalebox{0.9}{1-\textit{l}+}\frac{2}{\alpha}; \frac{1}{2}\scalebox{0.9}{-\textit{l+m}+}\frac{2}{\alpha}; \frac{\mp j m}{ q' \chi_k \Bar{P}_r(r)}\right)}{ \left(\pm j q' \chi_k \Bar{P}_r(r)/m\right)^{1-l}  \left(1\mp  j q' \chi_k \Bar{P}_r(r)/m\right)^{m-\frac{3}{2}}}\right)
    \right)\!\Bigg]_{r = r_0}^{r = \tau}\!\right)
    \end{multline}
}

\hrulefill
\end{figure*}
\end{theorem}
\begin{proof}
    The proof is provided in Appendix~\ref{sec:proofjoint1}.
\end{proof}

For the ease of analysis, we propose in Lemma~\ref{lem:cond_joint} the SCAIU conditioned on the AU's SINR. This metric allows one to answer questions such as "At places where coverage requirements are met, what is the probability to experience EMFE below $T_e$ at $x$ meters from an AU ?"
\begin{lemma}\label{lem:cond_joint}
    The joint CDF of the AU's SINR and IU's EMFE, conditioned on the AU's SINR being above threshold $T_c$ is
    \begin{equation}\label{eq:cond_joint}
        \mathcal{H}(T_{e}|T_{c};d) = \frac{\mathcal{J}(T_c, T_e;d)}{F_{\textrm{ cov}}\left(T_{c}\right)}.
    \end{equation}
\end{lemma}
\begin{proof}
    This metric is obtained from Bayes' rule.
\end{proof}

It is worth noting that the SCAIU conditioned on the IU's EMFE, $\mathcal{H}(T_{c}|T_{e};d)$could also be explored to address the question: "What is the best possible coverage that adheres to the IU's EMFE constraint?". Simple and tight bounds for $\mathcal{J}(T_{c}, T_{e};d)$ are the Fréchet lower bound (FLB) and the Fréchet upper bound (FUB), which are given in Lemma~\ref{lem:frechet}.
\begin{lemma}\label{lem:frechet} The expression of $\mathcal{J}(T_{c}, T_{e};d)$ in Theorem~\ref{th:joint1} is bounded by the Fréchet bounds
\begin{equation}
   \text{FLB} = \max\left(0, F_{\textrm{cov}}(T_{c}^{u}) + F_{\textrm{EMFE}}^{IU}(T_{e}) - 1\right)
\end{equation}and
\begin{equation}
   \text{FUB} = \min\left(F_{\textrm{cov}}(T_{c}^{u}), F_{\textrm{EMFE}}^{IU}(T_{e})\right) 
\end{equation}
such that $\text{FLB} \leq \mathcal{J}(T_{c}, T_{e};d) \leq \text{FUB}$.
\end{lemma}
Lemma~\ref{lem:frechet} can easily be adapted to obtain bounds on the conditional CDF~\eqref{eq:cond_joint} by dividing all members of the inequality by~$F_{\textrm{EMFE}}\left(T_{e}\right)$.

\section{Numerical Results}

In this section, the metrics derived in Section~\ref{sec:analytical_results} are analyzed using the values of network parameters listed in Table~\ref{tab:sim_param}. The bandwidth of 20~MHz is a typical value for the 3.5~GHz frequency band. Based on propagation models used in similar cellular networks \cite{GontierAccess}, we set the path loss exponent to $\alpha\! =\! 3.25$. The noise power is $\sigma^2 \!= \!10 \log_{10}(k\,T_0\,B_w) + 30 + \mathcal{F}_{dB}$ in dBm where $k$ is the Boltzmann constant, $T_0$ is the standard temperature (290~K), $B_w$ is the bandwidth and $\mathcal{F}_{dB} = \numprint[dB]{6}$ is the receiver noise figure \cite{10.5555/3294673}. The value of $k_{\textrm{max}} = 10$ was chosen to model a sufficient number of side lobes. Beyond this number, the side lobe power drops below one thousandth of the main lobe power, and the lobes no longer align with those of the theoretical pattern.

\begin{table}[h!]
    \begin{center}
    \caption{\label{tab:sim_param} Simulation parameters}
    \begin{tabular}{|c|c||c|c|} 
     \hline
     $f$ & \numprint[GHz]{3.5} & $\tau$ & \numprint[km]{3}\\ 
     $B_w$ & \numprint[MHz]{20} & $z$ & \numprint[m]{30}\\ 
     $r_e$ & \numprint[m]{0.3} & $P_t$ & \numprint[dBm]{48}\\
     $\lambda$ & \numprint[BS/km^2]{10} & $N$ & 64\\
     $\sigma^2$ & \numprint[dBm]{-95.40} & $\alpha$ & 3.25\\ 
     $k_{\textrm{max}}$ & 10 &  & \\ 
     \hline
    \end{tabular}
    \end{center}
    
\end{table}

\subsection{Comparison of Antenna Patterns}

The marginal distribution of EMFE for an IU located 10~m from the AU with $N \! = \! 64$ is illustrated in Fig.~\ref{fig:Comparison_patterns_expIU} for the different antenna patterns introduced in Subsection~\ref{ssec:gain}. Markers in the figure represent numerical values obtained from Theorem~\ref{th:exp_cdf}, while the solid line for the theoretical antenna model is derived from Monte Carlo simulations (MCSs). The close alignment between the CDFs obtained using the theoretical and multi-cosine antenna patterns validates the mathematical model. Comparatively, conventional antenna patterns from the literature, such as cosine, flat-top or Gaussian, exhibit larger errors of 36\%, 12 \% and 8\%, respectively, contrasting with an error of 2\% with the multi-cosine pattern. 


The SCAIU is showed again with a solid line in Fig.~\ref{fig:Comparison_patterns_expIU_S_and_I}. The CDF of the signal power from the AU's serving BS is represented by a dashed line, while the CDF of the interference power from all other BSs is depicted with a dash-dotted line. High EMFE values experienced by the IU primarily originate from the main lobe of the AU's serving BS. However, when the IU is not illuminated by this main lobe, the main contribution to its EMFE does not necessarily come from the side lobes of the AU's serving BS. This is evident because the CDF of signal power and total EMFE are not close for low EMFE values. Instead, the main lobes of other neighboring BSs contribute significantly, with side lobes of interfering BSs contributing to a lesser extent. This indicates that assessing EMFE solely from the main lobe of the closest BS is insufficient for a comprehensive characterization of EMFE. These observations are highly dependent on the values of $N$ and $d$, which are analyzed in Subsection~\ref{ssec:impact_parameters}.

Furthermore, the CCDF of SINR obtained in Lemma~\ref{lem:cov} is illustrated in Fig.~\ref{fig:Comparison_patterns_cov} for different antenna patterns, for $N \! = \! 64$ and $d \! = \! \numprint[m]{10}.$ 
The multi-cosine pattern achieves an error of less than 1\%, whereas the cosine, flat-top and Gaussian patterns show errors of 9\%, 15\% and 18\% respectively. Notably, the network scenario is not noise-limited, evident from the comparison between SINR and signal-to-noise ratio (SNR) CCDFs.

\begin{figure*}
\centering
\begin{minipage}{.32\linewidth}
  \centering
\includegraphics[width=0.95\linewidth, trim={5cm, 9cm, 5.9cm, 9cm}, clip]{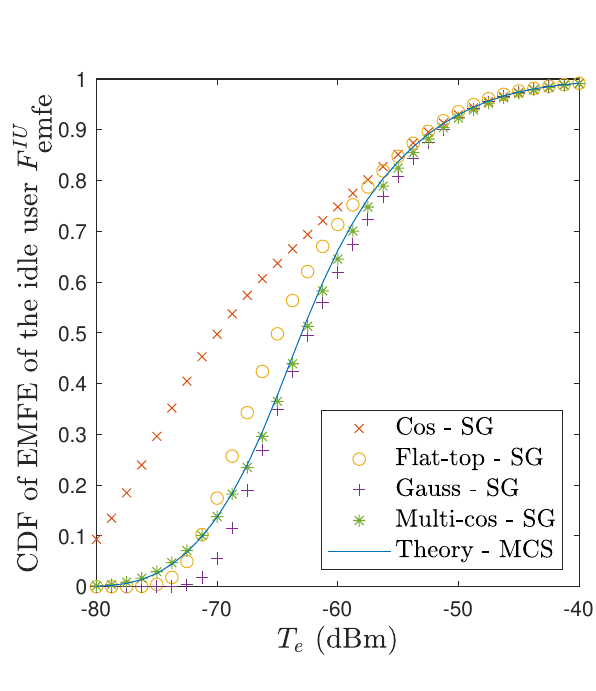}
\caption{CDF of IU's EMFE for different antenna patterns. $N \! = \! 64$, $d \! = \! \numprint[m]{10}$.}
\label{fig:Comparison_patterns_expIU}
\end{minipage}%
\hspace{0.cm}
\begin{minipage}{.32\linewidth}
\includegraphics[width=0.95\linewidth, trim={5cm, 9cm, 5.9cm, 9cm}, clip]{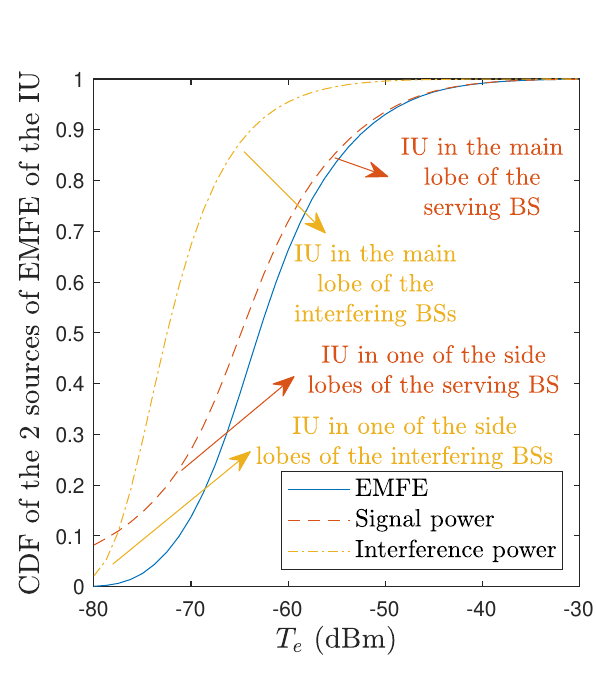} 
\caption{CDF of signal and interference power contributing to the total IU's EMFE. $N \! = \! 64$, $d \! = \! \numprint[m]{10}$.} 
\label{fig:Comparison_patterns_expIU_S_and_I}
\end{minipage}%
\hspace{0.cm}
\begin{minipage}{.32\linewidth}
    \centering
    \includegraphics[width=0.95\linewidth, trim={5cm, 9cm, 5.9cm, 9cm}, clip]{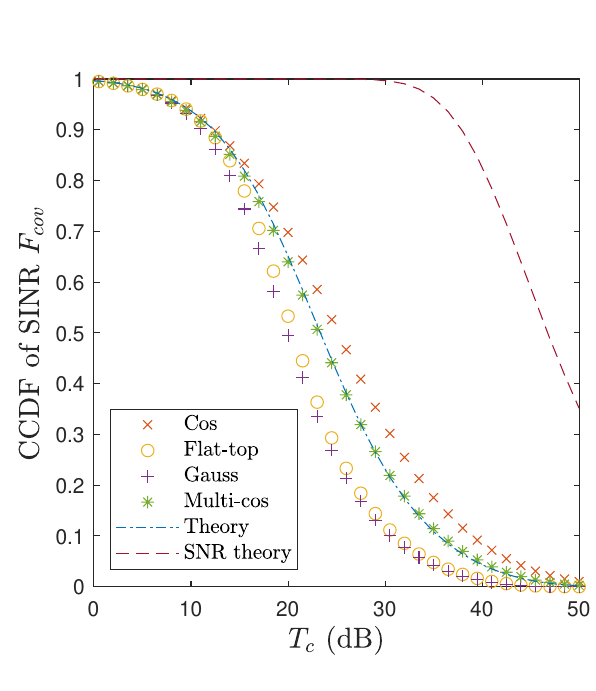}
    \caption{CCDF of SINR for different antenna patterns and zoom on the CCDF. $N \! = \! 64$, $d \! = \! \numprint[m]{10}$.}
    \label{fig:Comparison_patterns_cov}
\end{minipage}
\end{figure*}

\subsection{SCAIU}
\label{ssec:joint_analysis}
The SCAIU is depicted in Fig.~\ref{fig:BF_cond} for different values of $T_e$ and $d$, conditioned on an AU's SINR of 10~dB and $N \! = \! 64$. Markers are derived from Lemma~\ref{lem:cond_joint} while solid lines are obtained from MCSs. The proximity between the curves substantiates the accuracy of the mathematical expressions. Dash-dotted lines and dashed lines represent the FLB and FUB, respectively, given in Lemma~\ref{lem:frechet}. The Fréchet bounds provide an accurate approximation of the conditional CDF, aligning closely to the FUB at low EMFE limits and approaching the FLB at high ones. Observations reveal that $\mathcal{H}$ increases slowly with increasing $d$ when $T_e = \numprint[dBm]{-70}$, while it rises rapidly for higher EMFE limits. This behavior stems from the increased likelihood that the primary factor preventing the IU from remaining below the EMFE limit as it moves away is the main lobe of the AU's serving base station. This influence is depicted by the dotted line in Fig.~\ref{fig:BF_cond}, delineating the zone of influence of this lobe. At a distance of 10~m from the AU, this probability diminishes to 15\%. Consequently, at this distance, for $T_e = \numprint[dBm]{-55}$ or higher limits, $\mathcal{H}$ exceeds 82\% while it is only 16\% for $T_e = \numprint[dBm]{-40}$. This indicates that not only the main lobe of the AU's serving BS but also the side lobes and the lobes of other BSs prevent the IU's EMFE from falling below the limit. The conditional probability continues to rise as the IU moves away from the AU but eventually reaches a threshold. This threshold is reached sooner for high EMFE limits compared to lower ones. For $T_e = \numprint[dBm]{-70}$, this occurs around 60~m, with the probability that the EMFE caused by the main lobe of the AU's serving BS predominates being approximately 2\% at this distance.

\begin{figure}
\centering
\includegraphics[width=0.9\linewidth, trim={3cm, 9cm, 3cm, 9.5cm}, clip]{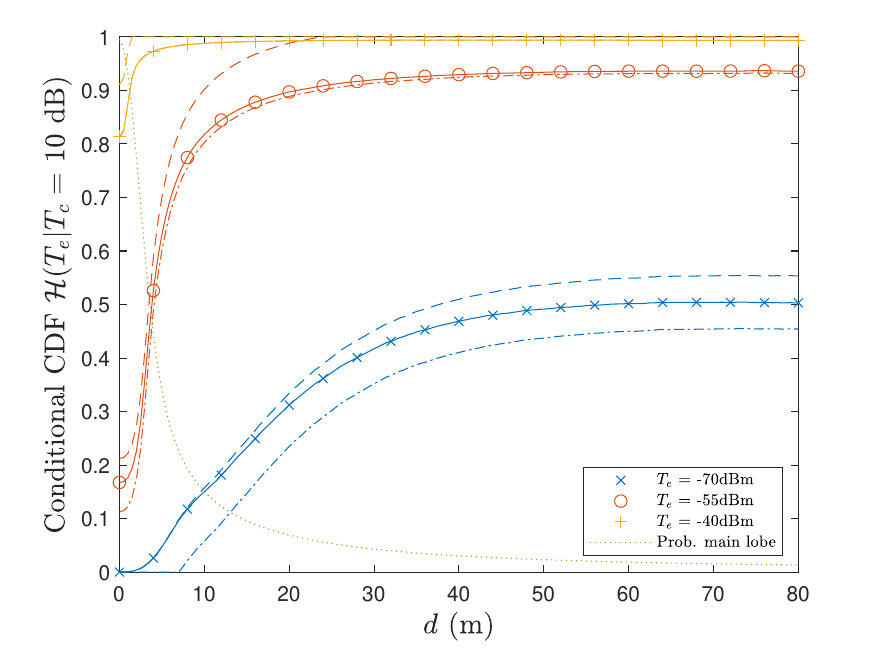}
\caption{SCAIU conditioned on an AU's SINR of 10~dB and $N \! = \! 64$. Markers are obtained from Lemma~\ref{lem:cond_joint} and solid lines from MCSs. The dash-dotted lines and dashed lines represent the FLB and FUB, respectively.}
\label{fig:BF_cond}
\end{figure}

\subsection{Impact of Number of Antenna Elements and Distance from the AU}
\label{ssec:impact_parameters}
The parameters $N$ and $d$ significantly influence the EMFE experienced by the IU, directly impacting the width of the main lobe and the  height and width of side lobes. Fig.~\ref{fig:BF_N_d_exp_countour} illustrates the probability $F^{IU}_{\textrm{EMFE}}(\numprint[dBm]{-40}) = x$ for fixed values $x$, across various values of $d$ and $N$. For example, consider $x = 0.95$. The contour plot indicates the pair of parameters $(d, N)$ where the median IU's EMFE is $\numprint[dBm]{-40}$. When the IU is very close to the AU, the main lobe should exhibit a highly directive beam. At 1~m, $N$ must exceed 128. At 10~m, $N \! = \! 16$ is enough and $N \! = \! 8$ is adequate at 30~m. Although these values of $N$ appear high, it is important to note the use of very strict limits for the needs of the analysis. According to the relationships in \cite{GontierTWC}, $T_e = \numprint[dBm]{-40}$ corresponds to $\numprint[\mu W/m^2]{171}$ or $\numprint[mV/m]{25}$. These thresholds are relatively low, reflecting a deliberate focus on limiting the EMFE of a single operator within a specific frequency band. Thus, these values cannot be directly compared with regulatory limits, which consider the total EMFE across all frequency bands and operators. Additionally, these lower values are justified by the fact that this study investigates the EMFE experienced by IUs, in contrast to most other studies in the literature, which primarily focus on the EMFE experienced by AUs.

\begin{figure}
\centering
\includegraphics[width=0.9\linewidth, trim={3cm, 9cm, 3cm, 9.5cm}, clip]{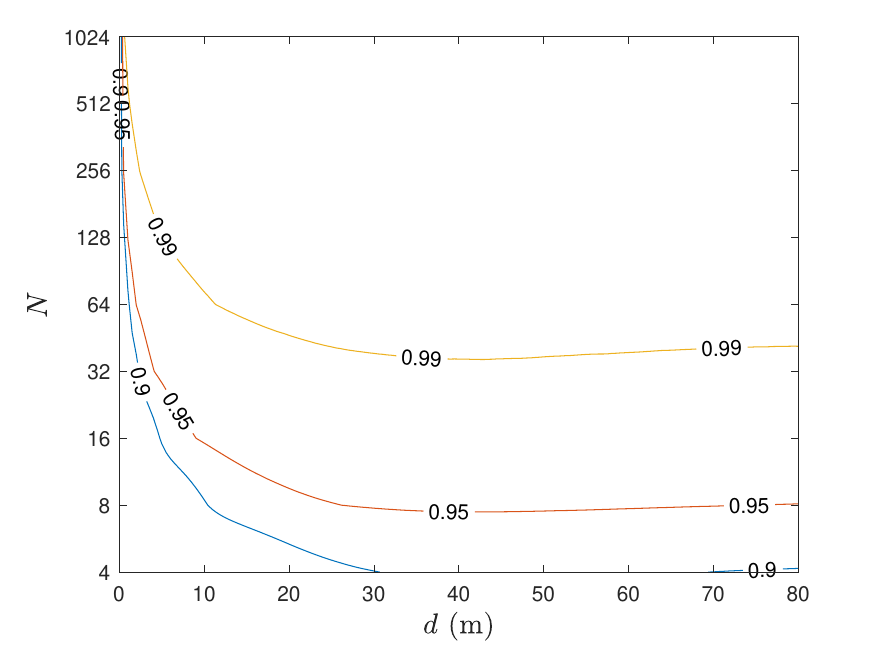}
\caption{Contour plot of $F^{IU}_{\textrm{EMFE}}(\numprint[dBm]{-40}) = x$ for various values of $N$, $d$ and $x$}
\label{fig:BF_N_d_exp_countour}
\end{figure}
A similar analysis can be performed for the SCAIU, as depicted in Fig.~\ref{fig:BF_N_d_joint_countour} with $T_c = \numprint[dB]{10}$ and $T_e = \numprint[dBm]{-40}$. Consider a potential future scenario that network providers might need to adhere to. For instance, suppose network providers must ensure that the EMFE of IUs located 2~m from an AU remains below -40~dBm at least 95\% of the time, while also maintaining the AU's SINR above 10~dB. According to Fig.~\ref{fig:BF_N_d_joint_countour}, this requirement can be met by employing directive antennas with at least 256~elements. 
\begin{figure}
\centering
\includegraphics[width=0.9\linewidth, trim={3cm, 9cm, 3cm, 9.5cm}, clip]{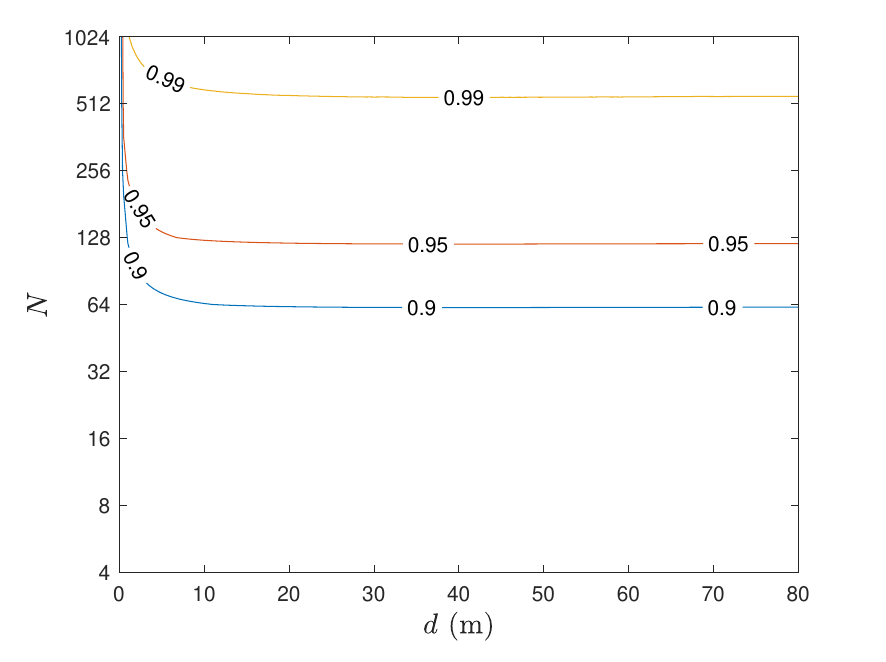}
\caption{Contour plot of the SCAIU for various values of $N$ and $d$, with $T_c = \numprint[dB]{10}$ and $T_e = \numprint[dBm]{-40}$.}
\label{fig:BF_N_d_joint_countour}
\end{figure}

\section{Conclusion}
\label{sec:conclusion}
In this paper, a marginal and joint analysis of the EMFE of an IU and the coverage experienced by an AU, located at a specific distance from the IU, was conducted. Due to the absence of sufficiently accurate and tractable antenna pattern models in the literature, a novel model, the multi-cosine model, was introduced. This model is significantly closer to the theoretical antenna pattern. The model presented in this paper addresses various questions regarding the IU's EMFE as a function of the number of antenna elements at the BS and the distance to the closest AU through the introduction of a new metric, the SCAIU. 

{\appendices

\section{Proof of the CF of the useful signal seen from the IU}
\label{sec:proofcfS0}
If the angle $\delta_0$ is larger than 60$^\circ$, the AU and the IU are not located in the same sector. The beam of the serving BS launched towards the AU has therefore no impact on the IU's EMFE. The CF of the useful signal from the point of view of the IU, conditioned on the distance location of the serving BS, should therefore consider these two cases in its definition:
\begin{equation}\label{eq:S_idle_proof}
\begin{split}
    \phi_S(q;d|X_0) = \mathbb E_{|h|}\left[e^{jq \Bar{P}_r(W_0) G(\delta_0) |h|^2} \mathds 1\left[|\delta_0| \leq \pi/3\right] \right]\\
    +\mathbb E_{|h|, \xi}\left[e^{jq \Bar{P}_r(W_0) G(\xi) |h|^2}\mathds 1\left[|\delta_0| >\pi/3\right] \right].
\end{split}
\end{equation}
where $\mathds 1[\cdot]$ is the indicator function. Applying the expectation operator $\mathbb E_{|h|}[\cdot]$ on the term for the case $\mathds 1\left[|\delta_0| \leq \pi/3\right]$, we obtain
\begin{equation}\label{eq:Phi_S_nogain}
    \mathbb E_{|h|}\left[e^{jq \Bar{P}_r(W_0) G(\delta_0) |h|^2} \right] = \left(1-jq \Bar{P}_r(W_0) G(\delta_0)/m\right)^{-m}.
\end{equation}
For the term with the case $\mathds 1\left[|\delta_0| >\pi/3\right] = 1-\mathds 1\left[|\delta_0| \leq \pi/3\right]$, let us develop the Taylor series of the exponential function. Then, the infinite sum and the expectation operators can be swapped. Knowing that $\mathbb E_{|h|}\left[|h|^{2k}\right] = \frac{\Gamma(m+k)}{\Gamma(m) m^k}$, using the notation $z = jq \Bar{P}_r(W_0)$, we have
\begin{align}\label{eta_S_def}
\begin{split}
    &\mathbb E_{|h|, \xi}\left[e^{z G(\xi) |h|^2}\right] = \mathbb E_{|h|, \xi}\left[\sum\limits_{k = 0}^\infty \frac{\left(z \,G(\xi) |h|^2\right)^k}{k!}\right]\\
    &= 1+\sum\limits_{k = 1}^\infty \frac{z^k}{k!} \mathbb E_{\xi}\left[G^k(\xi)\right] \frac{\Gamma(m+k)}{\Gamma(m) m^k} := \eta_S(q|X_0).
\end{split}
\end{align}
\eqref{eq:cfSidle} is then obtained by inserting  \eqref{eq:Phi_S_nogain} and \eqref{eta_S_def} in \eqref{eq:S_idle_proof}. Then $\eta_S(q|X_0)$ can be developed separately for each gain function, using Proposition~\ref{prop:gain}:
\begin{itemize}
    \item \eqref{eq:eta_S_ft} is obtained by inversely applying the definition of the binomial series.
    \item \eqref{eq:eta_S_mc} is obtained by inversely applying the definition of the series expansion of the hypergeometric function ${}_2F_1(\cdot;\cdot;\cdot)$.
    \item The case of the Gaussian pattern is more complex. In that case, using \eqref{eq:mom_GG}, $\eta_S(q|X_0)$ is given by
\end{itemize}
    {\small
    \begin{equation}\label{eta_S_G_proof}
    \begin{split}
        &\eta_S(q|X_0) = \sum\limits_{k = 0}^\infty \left(\frac{z\,g}{m}\right)^k  \frac{\Gamma(m+k)}{\Gamma(m) k!}\\
        &+ \sum\limits_{k = 1}^\infty \frac{3}{2}\frac{z^k}{k!} \frac{\Gamma(m+k)}{\Gamma(m) m^k} \sum\limits_{p = 1}^k \! \binom{k}{p} (N-g)^p g^{k-p} \frac{\erf\left(\tfrac{\pi \sqrt{p \eta}}{3}\right)}{\sqrt{\pi p \eta}}
    \end{split}
    \end{equation}}
The first term in \eqref{eta_S_G_proof} is solved by inversely applying the definition of the binomial series. The second term can be rewritten as
\begin{align}\label{eta_S_G_proof2}
    \begin{split}
        &\sum\limits_{p=1}^\infty  \frac{3\erf\left(\tfrac{\pi \sqrt{p \eta}}{3}\right)}{2 \sqrt{\pi p \eta}} (N-g)^p \sum\limits_{k=p}^\infty \!\frac{z^k}{k!} \frac{\Gamma(m+k) }{\Gamma(m) m^k } \binom{k}{p} g^{k-p}\\
        &= \sum\limits_{p=1}^\infty  \frac{3\erf\left(\tfrac{\pi \sqrt{p \eta}}{3}\right)}{2 \sqrt{\pi p \eta}} (N-g)^p \frac{z^p}{p!}\frac{ \Gamma(m+p) }{\Gamma(m) m^p } \left(1-\frac{z}{m}\right)^{-(m+p)}
    \end{split}
\end{align}
Replacing \eqref{eta_S_G_proof2} in \eqref{eta_S_G_proof} gives \eqref{eq:eta_S_G} in Proposition \ref{prop:cf_S}.
    
\section{Proof of Proposition~\ref{prop:cf}}
\label{sec:proofcf}
The CF function of the interference of the AU at the origin is defined and commonly written as $\phi_{I}(q|r_0) = \mathbb E_{I_0}\left[\exp(jq {I_0}(0))\right]$. Following from \cite{yu2017}, and contrarily to the conventional procedure, the first step consists of taking first the expectation over the interferers' locations, by means of the probability generating functional:
\begin{equation}
\begin{split}\label{eq:phiI0first}
     &\phi_{I}(q|R_0) = \mathbb E_{I_0}\left[e^{jq {I_0}(0)}\right]\\
     &\quad \! = \! \exp\left(-{\pi\lambda} \mathbb E_{\xi, |h|}\left[2\int_{R_0}^\tau \!  \left(1- e^{j q {P}_{r}(r)} \right)r \,dr\right]\right)\\
     &\quad \! = \! \exp\Bigg(-{\pi\lambda} \Bigg(\underbrace{\tau^2-R_0^2 - \mathbb E_{\xi, |h|}\left[\int_{R_0}^\tau \!  2 e^{j q {P}_{r}(r)} r \,dr\right]}_{\eta_I(q|R_0)}\Bigg)\Bigg).
\end{split}
\end{equation} 
By using the change of variable $t \to -j q {P}_{r}(r)$ and writing $\delta = 2/\alpha$, the integral can be rewritten
\begin{equation}\label{eq:temp1}
\begin{split}
    &\int_{R_0}^\tau \!  2 e^{j q {P}_{r}(r)} r \,dr \\
    &\quad= \delta\int_{-j q {P}_{r}(R_0)}^{-j q {P}_{r}(\tau)} \!\left(j q P_t \kappa^{-1} G  |h|^2\right)^{\delta} e^{-t} (-t)^{-1-\delta} \,dt\\
    &\quad=\delta \left(-j q P_t \kappa^{-1} G |h|^2\right)^\delta\left[\Gamma(-\delta,-j q P_r(r)\right]_{r = R_0}^{r=\tau}.
\end{split}
\end{equation}
$\Gamma(a, z)$ is the upper incomplete Gamma function whose definition and expansion series are 
\begin{equation}
    \Gamma(a, z) = \int_z^{\infty} {e^{-t}}{t^{a-1}}dt = \Gamma(a) - \sum\limits_{k=0}^{\infty} \frac{(-1)^kz^{a+k}}{k! (a+k)}.
\end{equation}
Using this expansion series in \eqref{eq:temp1} and inserting it in $\eta_I(q|R_0)$ in \eqref{eq:phiI0first} and letting $\Xi(r) = j q \Bar{P}_r(r)$ gives
{\small
\begin{equation*}
    \eta_I(q|R_0)= \left[r^2+\delta\left(\tau^2+z^2\right) \sum\limits_{k=0}^{\infty} \frac{\Xi^k(r) \mathbb E_{\xi, |h|}\!\left[G^k |h|^{2k}\right]}{k! (k-\delta)}\right]_{r = R_0}^{r=\tau}\!.
\end{equation*}}
Extracting the terms $k=0$ and using $\mathbb E\left[|h|^{2k}\right] = \frac{\Gamma(m+k)}{\Gamma(m) m^k}$, gives after some simplifications
{\small
\begin{equation}\label{eta_I}
        \eta_I(q|R_0) = \delta\left[\left(r^2+z^2\right)\sum_{k=1}^\infty\frac{ \mathbb E_{\xi}\!\left[G^k(\xi)\right] \Xi^k(r)}{k!(k-\delta)} \frac{\Gamma(m+k)}{\Gamma(m) m^k} \right]_{r = R_0}^{r=\tau}\!.
\end{equation}}
Similarly to what is done for $\eta_S(q|X_0)$ in Appendix~\ref{sec:proofcfS0}, $\eta_I(q|R_0)$ can be developed separately for each gain function, using Proposition~\ref{prop:gain}. The expressions of $\eta_I(q|R_0)$ for the flat-top and multi-cos patterns in Proposition~\ref{prop:cf} are obtained by using the series expansion of the generalized hypergeometric function $_pF_q(\cdot;\cdot)$. For case of the Gaussian pattern is again more complex. Using \eqref{eq:mom_GG} and writing $\eta_I(q|r)$ such that  $\eta_I(q|R_0) = \left[\eta_I(q|r)]\right]_{r = R_0}^{r=\tau}$, we have
    {\footnotesize
    \begin{equation}\label{eta_I_G_proof}
    \begin{split}
        &\eta_I(q|r) = \delta\left(r^2+z^2\right) \sum\limits_{k = 0}^\infty \left(\frac{\Xi(r)g}{m}\right)^k  \!\frac{\Gamma(m+k)}{\Gamma(m) k! (k-\delta)}+r^2+z^2\\
        &+ \sum\limits_{k = 1}^\infty \frac{3}{2}\frac{\Xi^k(r)\delta\left(r^2+z^2\right)\Gamma(m+k)}{\Gamma(m) m^k(k-\delta)k!} \sum\limits_{p = 1}^k \! \binom{k}{p} (N-g)^p g^{k-p} \frac{\erf\left(\tfrac{\pi \sqrt{p \eta}}{3}\right)}{\sqrt{\pi p \eta}}
    \end{split}
    \end{equation}}
The first term in \eqref{eta_I_G_proof} is solved by inversely applying the series expansion of ${}_2F_1(\cdot;\cdot;\cdot)$:
\begin{equation*}
    \sum\limits_{k = 0}^\infty \left(\frac{\chi\,g}{m}\right)^k  \frac{\Gamma(m+k)}{\Gamma(m) k! (k-\delta)} = -\frac{1}{\delta} {}_2F_1(-\delta,m;1-\delta;\frac{\chi g}{m}).
\end{equation*}
Introducing the generalized beta function $B(\cdot;\cdot,\cdot)$, the term of the second line in \eqref{eta_I_G_proof} can be rewritten as
\begin{align}\label{eta_I_G_proof2}
    \begin{split}
        &\sum\limits_{p=1}^\infty  \frac{3\erf\left(\tfrac{\pi \sqrt{p \eta}}{3}\right)}{2 \sqrt{\pi p \eta}} (N-g)^p \sum\limits_{k=p}^\infty \!\frac{z^k}{k!} \frac{\Gamma(m+k) }{\Gamma(m) m^k (k-\delta)} \binom{k}{p} g^{k-p}\\
        &= \sum\limits_{p=1}^\infty  \frac{3\erf\left(\tfrac{\pi \sqrt{p \eta}}{3}\right)}{2 \sqrt{\pi p \eta}} \left(\frac{N}{g}-1\right)^p \frac{\Gamma(p+m)}{\Gamma(m)p!}\\
        &\qquad \times B\left(\frac{gz}{m};p-\delta,-p+1-m\right)\left(\frac{gz}{m}\right)^\delta.
    \end{split}
\end{align}

\section{Proof of Theorem~\ref{th:joint1}}
\label{sec:proofjoint1}
\subsection{General Form of the Metric}
\label{ssec:general_form}
The fading coefficients affecting the links related to the two locations being independent, conditioned on the PP $\Psi$, the joint metric can be decomposed as $\mathcal{J}(T_{c}, T_{e};d) = \mathbb E_{\Psi}\left[\mathbb{P}\left[\textrm{\normalfont SINR}_0 > T_{c}|\Psi\right]\,\mathbb{P}\left[ \mathcal{P}^{IU}(d) < T_{e}|\Psi\right]\right]$.
The two factors in the above product can be developed using the Gil-Pelaez theorem. Using again the assumption of a CF of the interference identical for the AU and the IU, let $\phi_{\text{\normalfont{SINR}}}(q, T_c|\Psi)=\phi_{S}(q; 0 |\Psi)\,\phi_{I}(-q T_{c}|\Psi)\, e^{-j q T_{c} \sigma^2}$ and $\phi_E^{IU}(q;d|\Psi) = \phi_{S}(q; d|\Psi)\,\phi_{I}(q|\Psi)\, e^{-j q T_{e}}$ be respectively the CFs of the signal and interference conditioned on $\Psi$. Using these notations, we get
\begin{equation*}
\begin{split}
    \mathcal{J}(T_{c}, T_{e};d) = \mathbb E_{\Psi}\left[\left(\frac{1}{2}+ \int_0^\infty \textrm{\normalfont Im}\!\left[\phi_{\text{\normalfont{SINR}}}(q, T_c|\Psi)\right] \frac{1}{\pi q} dq\right)\right.\\
    \left.\left(\frac{1}{2}- \int_0^\infty \textrm{\normalfont Im}\!\left[\phi_E^{IU}(q';d|\Psi)\, e^{-j q' T_{e}}\right] \frac{1}{\pi q'} dq\right)\right].
\end{split}
\end{equation*}
By distributing the terms, then, the expectation operator, we get
{\footnotesize
\begin{equation*}
\begin{split}
    \mathcal{J}(T_{c}, T_{e};d) = -\frac{1}{4} + \frac{1}{2}\,\mathbb E_{\Psi}\Big[\underbrace{\frac{1}{2}+ \int_0^\infty \textrm{\normalfont Im}\!\left[\phi_{\text{\normalfont{SINR}}}(q, T_c|\Psi)\right] \frac{1}{\pi q}\, dq}_{F_{\textrm{cov}}(T_{c}|\Psi)}\Big]\\
    +\frac{1}{2}\,\mathbb E_{\Psi}\Big[\underbrace{\frac{1}{2}- \int_0^\infty \textrm{\normalfont Im}\!\left[\phi_E^{IU}(q';d |\Psi)\, e^{-j q' T_{e}}\right] \frac{1}{\pi q'}\, dq'}_{F_{\textrm{EMFE}}^{IU}(T_{e}; d|\Psi)}\Big]\\
    -\frac{1}{\pi^2}\!\underbrace{\mathbb E_{\Psi}\Big[\int_0^\infty \!\textrm{\normalfont Im}\!\left[\phi_{\text{\normalfont{SINR}}}(q, T_c|\Psi)\right] \! \frac{dq}{q} \! \int_0^\infty \!\textrm{\normalfont Im}\!\left[\phi_E^{IU}(q';d |\Psi)\, e^{-j q' T_{e}}\right]\! \frac{dq'}{q'}\Big]}_{\Upsilon(T_{c}, T_{e}; d)}.
\end{split}
\end{equation*}}
By applying the expectation operator, one has $\mathbb E_{\Psi}\left[F_{\textrm{cov}}(T_{c}|\Psi)\right] = F_{\textrm{cov}}(T_{c})$ and $\mathbb E_{\Psi}\left[F_{\textrm{EMFE}}^{IU}(T_{e}; d|\Psi)\right] = F_{\textrm{EMFE}}^{IU}(T_{e}; d)$.
due to the motion-invariance of the H-PPP in $\mathbb{R}^2$, which gives the first line of Theorem~\ref{th:joint1}.

\subsection{Decomposition of $\Upsilon(T_{c}, T_{e}; d)$}
\label{ssec:decomposition}
The expectation over $\Psi$ in the last expression of $\Upsilon(T_{c}, T_{e}; d)$ can be decomposed in the following manner:
\begin{equation*}
    \mathbb E_{\Psi}\left[\cdot\right] \to \mathbb E_{X_0}\left[\mathbb E_{\Psi \backslash \{X_0\}}\left[\cdot\right]\right]
\end{equation*}
where the coordinates of the serving BS $X_0$ are $(R_0, \Theta_0)$. Additionally, we write $\Psi^*=\Psi \backslash \{X_0\}$. Building up on these notations and using the PDF~\eqref{eq:pdf_R0}, we obtain
\begin{equation*}
    \Upsilon(T_{c}, T_{e}; d) = \frac{1}{2\pi}\int_{r_e}^{\tau}\int_0^{2\pi}\Upsilon(T_{c}, T_{e}; d| r_0, \theta_0)\,d\theta \,f_{R_0}(r_0)\,dr_0.
\end{equation*}
where
\begin{equation*}\begin{split}
    &\Upsilon(T_{c}, T_{e}; d| R_0, \Theta_0) \\
    &= \mathbb E_{\Psi^*}\Big[ \int_0^\infty \textrm{\normalfont Im}\Big[\phi_{I}(-q T_{c}|\Psi)\,\phi_{S}(q; 0 |\Psi)\, e^{-j q T_{c} \sigma^2}\Big] \frac{1}{q} dq \Big.\\
    &\qquad \Big.\times \int_0^\infty \textrm{\normalfont Im}\Big[\phi_{I}(q'|\Psi)\,\phi_{S}(q'; d|\Psi)\, e^{-j q' T_{e}}\Big] \frac{1}{q'} dq'\Big]
\end{split}\end{equation*}
Let $v_1\! = \!\phi_{S}(q; 0 |\Psi)\, e^{-j q T_{c} \sigma^2}$ and $v_2 \!=\! \phi_{S}(q'; d|\Psi)\, e^{-j q' T_{e}}$. By swapping the expectation and the integrals, following Fubini's theorem, we obtain
\begin{equation}\begin{split}
    &\Upsilon(T_{c}, T_{e}; d| R_0, \Theta_0) \\
    &= \int_0^{\infty} \!\int_0^{\infty} \!\underbrace{\mathbb E_{\Psi^*}\!\left[\textrm{\normalfont Im}\!\left[\phi_{I}(-q T_{c}|\Psi)v_1\right] \textrm{\normalfont Im}\!\left[\phi_{I}(q'|\Psi)v_2\right]\right]}_{\epsilon(q, q';T_{c},T_{e},d|R_0, \Theta_0)}\!\frac{dq\, dq'}{q \,q'}.
\end{split}\end{equation}
\subsection{Decomposition of $\epsilon(q, q';T_{c},T_{e},d|r_0, \theta_0)$}
\label{ssec:decomposition3}
Since $\phi_{S}$ does not depend on $\Psi^*$, by using $\textrm{\normalfont Im}[x] = \frac{x-\Bar{x}}{2}$ and $\text{Re}[x] = \frac{x+\Bar{x}}{2}$, we obtain
\begin{multline}
    \epsilon(q, q';T_{c},T_{e},d|R_0, \Theta_0) \\
    = \frac{1}{4}\,\mathbb E_{\Psi^*}\!\left[\left(\phi_{I}\left(-q T_{c}| \Psi\right)\,v_1-\Bar{\phi_{I}}\left(-q T_{c}| \Psi\right)\,\Bar{v_1}\right)\right. \\
    \qquad \times \,\left.\left(\phi_{I}\left(q'| \Psi\right)\,v_2-\Bar{\phi_{I}}\left(q'| \Psi\right)\,\Bar{v_2}\right)\right]\\
    = \frac{1}{2}\,\text{Re}\left[\gamma_+(q, q'|r_0)\,v_1\,v_2\right]-\frac{1}{2}\,\text{Re}\left[\gamma_-(q, q'|r_0)\,v_1\,\Bar{v_2}\right]
\end{multline}where we define $
    \gamma_+(q, q'|r_0) = \mathbb E_{\Psi^*}\!\left[\phi_{I}\left(q| \Psi\right)\,\phi_{I}\left(q'| \Psi\right)\right] $
and $\gamma_-(q, q'|r_0) = \mathbb E_{\Psi^*}\!\left[\phi_{I}\left(q| \Psi\right)\,\Bar{\phi_{I}}\left(q'| \Psi\right)\right]$.

\subsection{Decomposition of $\gamma_+$ and $\gamma_-$}
\label{ssec:decomposition4}
The methodology employed to calculate the expressions of $\gamma_+$ and $\gamma_-$ aligns with the approach detailed in Appendix~C of \cite{GontierMeta} for the identical network pertaining to a RU. The only modification required lies in adjusting the lower bound of the integral in equation~(30), substituting $r_0$ for $r_e$.

}

\bibliographystyle{IEEEtran}
\bibliography{bibli}

\vfill

\end{document}